\documentclass[abstract=true, DIV=14]{scrartcl}

\usepackage{amsmath,amssymb,amsthm}
\usepackage[mathscr]{euscript}
\usepackage[shortlabels]{enumitem}
\usepackage{graphicx}
\usepackage{subcaption}
\usepackage{float}
\usepackage{xcolor}
\usepackage{microtype}
\usepackage[T1]{fontenc}
\usepackage{nicefrac,xfrac}
\usepackage[english]{babel}
\usepackage{upgreek}
\usepackage[colorlinks]{hyperref}
\usepackage[capitalize]{cleveref}
\usepackage{authblk}
\usepackage{tikz}
\usepackage{todonotes}
\usepackage{changepage}
\usetikzlibrary{calc,positioning}

\makeatletter\input{t1cmss.fd}\makeatother
\DeclareFontShape{T1}{cmss}{bx}{n}%
  {<5><6><7><8><9><10->ecsx1000}{}

\newcommand{\ceil}[1]{\lceil #1 \rceil}

\newcommand{\key}{\mathrm{key}}
\newcommand{\col}{\mathrm{col}}
\newcommand{\wit}{\mathrm{wit}}
\renewcommand{\alpha}{\upalpha}

\setkomafont{captionlabel}{\bfseries}
\renewcaptionname{english}{\figurename}{Fig.}
\setcapindent{0pt}
\newtheorem{theorem}{Theorem}[section]
\newtheorem{lemma}[theorem]{Lemma}

\newtheorem{corollary}[theorem]{Corollary}

\theoremstyle{definition}

\title{An optimal deterministic algorithm for finding a strict saddlepoint}

\author{Justin Dallant\thanks{Email: \href{mailto:justin.dallant@tu-dresden.de}{justin.dallant@tu-dresden.de}}}
\affil{Faculty of Computer Science, TU Dresden, Germany}

\date{}

\begin{document}

\maketitle

\begin{abstract}
Given an $n\times n$ matrix $A$, a \emph{saddlepoint} of $A$ is an entry that is the maximum in its row and the minimum in its column. It is a \emph{strict saddlepoint} if no other entry in its row or column has the same value.

Finding a non-strict saddlepoint requires $\Theta(n^2)$ matrix queries in the worst case. In contrast, a strict saddlepoint can be found with only $O(n)$ queries. In 1991, Bienstock, Chung, Fredman, Schäffer, Shor, and Suri---and, independently, Byrne and Vaserstein---showed that one can find a strict saddlepoint (or certify that none exists) in $O(n\log n)$ time using $O(n)$ matrix queries. In 2024, Dallant, Haagensen, Jacob, Kozma, and Wild gave an $O(n\log^* n)$-time algorithm, followed shortly after by an optimal \emph{randomized} algorithm running in $O(n)$ time with high probability. Whether $O(n)$ time could also be achieved deterministically was left open by these works.

Here we resolve this question by presenting a simple \emph{deterministic} algorithm that finds a strict saddlepoint, or reports that none exists, in optimal $O(n)$ time. Our algorithm combines elementary ingredients from previous approaches with linear-time selection from a collection of sorted lists.
\end{abstract}

\medskip
{\paragraph*{Acknowledgements} \small I thank L{\'{a}}szl{\'{o}} Kozma for useful comments on an earlier draft of this paper.}

\section{Introduction}
Given a matrix $A$ with entries from an ordered set, a \emph{saddlepoint} of $A$ is an entry that is a maximum in its row and a minimum in its column. A \emph{strict saddlepoint} is strictly larger than every other entry in its row and strictly smaller than every other entry in its column. In game-theoretic terms, a saddlepoint is a pure-strategy Nash equilibrium of the two-player zero-sum game whose payoff matrix is $A$.

This seemingly small distinction between strict and non-strict saddlepoints has significant algorithmic consequences. In the non-strict case, finding a saddlepoint, or even deciding whether one exists, may require inspecting $\Omega(n^2)$ entries of an $n\times n$ matrix \cite{LlewellynToveyTrick}; this lower bound also holds in expectation for randomized algorithms \cite{DallantEtAlRand}. The simple quadratic-time algorithms described by Knuth are thus worst-case optimal \cite[\S~1.3.2]{Knuth}. A matrix may also have several non-strict saddlepoints, although they all have the same value.

A strict saddlepoint, by contrast, is unique and can be found after querying a vanishingly small fraction of all matrix entries. Llewellyn, Tovey, and Trick first obtained a subquadratic algorithm with running time $O(n^{\log_2 3})$, where $\log_2 3\approx 1.585$ \cite{LlewellynToveyTrick}.
In 1991, Bienstock, Chung, Fredman, Sch\"affer, Shor, and Suri \cite{BienstockEtAl}---and, independently, Byrne and Vaserstein \cite{ByrneVaserstein}---gave deterministic $O(n\log n)$-time algorithms using $O(n)$ queries. Their running time is dominated by sorting or by maintaining a sorted set under a linear number of updates.

In 2024, Dallant, Haagensen, Jacob, Kozma, and Wild reduced the deterministic running time to $O(n\log^* n)$ \cite{DallantEtAlDet}. Their approach finds a relaxed object called a pseudo-saddlepoint by recursive decomposition and then checks whether an entry of the same value is a strict saddlepoint. The same authors subsequently gave a Las Vegas randomized algorithm whose running time is $O(n)$ with high probability \cite{DallantEtAlRand}. Whether the same bound can be achieved deterministically was left open.

We answer this question affirmatively, thus completely settling the complexity status of the problem.

\begin{theorem}\label{thm:main}
Given an $n\times n$ matrix $A$, we can identify its strict saddlepoint, or report that none exists, deterministically in $O(n)$ worst-case time.
\end{theorem}

The running time is optimal up to a constant factor: verifying a proposed strict saddlepoint requires comparing it with the other $n-1$ entries of its row and the other $n-1$ entries of its column. Our algorithm is comparison-based and assumes only that any pair of matrix entries can be compared in constant time.

The only non-elementary black box we use is the following standard consequence of linear-time heap selection \cite{Frederickson}; see also the simplified algorithm of Kaplan et al.\ \cite[Theorem~8]{KaplanEtAl} and the related work of Frederickson and Johnson \cite{FredericksonJohnson}.

\begin{lemma}\label{lem:list-selection}
There is a deterministic algorithm that, given $q$ nondecreasing lists and an integer $k$, identifies the $k$ smallest items in their union in $O(q+k)$ time, accessing every list from left to right (i.e.\ no element is ever accessed before its left neighbor). In particular, there is an absolute constant $\kappa\geq 1$ such that the algorithm queries at most $\kappa(q+k)$ list entries in total.
\end{lemma}

The starting point for our algorithm is the simpler $O(n\log\log n)$ algorithm from \cite[\S~6]{DallantEtAlDet}. That algorithm has two phases. In the first phase, it repeatedly applies a one-sided reduction; each application deletes a constant fraction of the rows or columns in linear time. This continues until the shorter side of the matrix has length at most $n/\log n$. In the second phase, the algorithm reduces the longer side---say, the width---by discarding columns through repeated heap operations. Once both dimensions are $O(n/\log n)$, it invokes the known $O(N\log N)$-time algorithm on a square matrix of side length $N=O(n/\log n)$. This final step takes $O(n)$ additional time.

Our algorithm keeps the first phase, but changes both the stopping rule and the second phase. We stop the first phase after only a constant number of reductions. At that point, either both dimensions have decreased by a constant factor, or one dimension is a sufficiently small constant fraction of the other. In the latter case, rather than repeatedly performing heap operations to reduce the longer dimension, we define one sorted list for each row or column, depending on which dimension is longer, and process all lists simultaneously using linear-time selection. This removes a constant fraction of the longer dimension in one linear-time step. Repeating this constant-factor reduction yields a geometric series and therefore an overall linear running time.

We recall some preliminaries and useful results from previous work in \S~\ref{sec:prelim}. We then formalize the reduction of the longer dimension in \S~\ref{sec:thin} and complete the algorithm in \S~\ref{sec:overall}.

\section{Preliminaries}\label{sec:prelim}

In this section, we recall basic facts and auxiliary results from previous work, including proofs for completeness.

For a positive integer $n$, let $[n]=\{1,\ldots,n\}$. Let $A$ be an $h\times w$ matrix, and write $A_{i,j}$ for its entry in row $i$ and column $j$.

An entry $A_{i,j}$ is a \emph{saddlepoint} if
\[
A_{i,j}\geq A_{i,j'}\quad\text{for all }j'\ne j,
\qquad\text{and}\qquad
A_{i,j}\leq A_{i',j}\quad\text{for all }i'\ne i.
\]
It is a \emph{strict saddlepoint} if all these inequalities are strict. If a matrix has a strict saddlepoint, it is unique. Throughout the paper, \emph{saddlepoint} will mean strict saddlepoint unless explicitly stated otherwise.

\subsection{Distinct keys and deleting rows/columns}

It is convenient to assume that all matrix entries are distinct. This is without loss of generality. Replace each entry $A_{i,j}$ by
\[
\widehat A_{i,j}=(A_{i,j},i,j),
\]
ordered lexicographically. The entries of $\widehat A$ are pairwise distinct, and every strict saddlepoint of $A$ is also a strict saddlepoint of $\widehat A$. We run the algorithm on $\widehat A$ and verify the reported entry in $A$ at the end. Thus, for the remainder of the analysis, we assume that all entries are distinct.

We repeatedly pass to submatrices obtained by deleting rows and columns. We always maintain the following invariant:

\begin{quote}
If the original matrix has a saddlepoint, the current submatrix contains the same saddlepoint.
\end{quote}

Deleting a row or column that does not contain the saddlepoint preserves the invariant. We will also occasionally restore previously deleted rows or columns, to ease the exposition. This does not break the invariant, because a saddlepoint of the original matrix cannot be invalidated by restoring rows or columns of that matrix.

Deletions and restorations are performed implicitly rather than by maintaining an explicitly modified copy of the matrix. We maintain arrays containing the indices of the current rows and columns and use them to query the original matrix. This bookkeeping is straightforward because the arrays need not remain ordered: permuting the rows or columns of a matrix does not affect the existence of saddlepoints.

\subsection{A threshold test}

The following staircase test is a streamlined version of the feasibility test used in \cite{DallantEtAlDet}.

\begin{lemma}[\cite{DallantEtAlDet}]\label{lem:threshold}
Let $A$ be an $m\times n$ matrix with distinct entries, and let $s$ be a specified value. In $O(m+n)$ time, we can return one of the following two correct outcomes:
\begin{enumerate}[(i)]
    \item the saddlepoint of $A$, if it exists, has value greater than $s$;
    \item the saddlepoint of $A$, if it exists, has value at most $s$.
\end{enumerate}
\end{lemma}

\begin{proof}
Perform the following search in the matrix, starting at the top-left position $(1,1)$. At an entry $q$, move down if $q>s$, and move right otherwise. If the search exits the matrix through the bottom, then it encountered an entry greater than $s$ in every row. Hence every row maximum is greater than $s$. Thus the saddlepoint value, if one exists, is greater than $s$; return (i).

Otherwise, the search exits through the right boundary. It then encountered an entry at most $s$ in every column, so every column minimum is at most $s$. Thus the saddlepoint value, if one exists, is at most $s$; return (ii).
\end{proof}

\subsection{A one-sided reduction}

We next extract the part of the earlier $O(n\log\log n)$ approach that we need.

\begin{lemma}[\cite{DallantEtAlDet}]\label{lem:one-sided}
Let $A$ be an $m\times n$ matrix with distinct entries, where $m,n\geq 4$. In $O(m+n)$ time, we can produce a submatrix preserving every saddlepoint of $A$ by deleting at least $m/4$ rows or at least $n/4$ columns.
\end{lemma}

\begin{proof}
Suppose $m\geq n$.
For each row $i\in[m]$, choose the entry
\[
d_i=A_{i,\phi(i)},
\qquad
\phi(i)=\ceil{\frac{in}{m}}.
\]
Every column contains at most $\ceil{m/n}$ of the chosen entries. Compute the median $s$ of $d_1,\ldots,d_m$ in $O(m)$ time and run the threshold test of Lemma~\ref{lem:threshold}.

Suppose first that a saddlepoint, if it exists, must be greater than $s$ (outcome (i) of the test). Delete every column containing a chosen entry $d_i\leq s$. Such a column cannot contain the saddlepoint: the saddlepoint is a minimum in its column, so its value would be at most $d_i\leq s$. At least $m/2$ chosen entries are at most $s$, and every column contains at most $\ceil{m/n}\leq 2m/n$ chosen entries. Hence we delete at least
\[
\frac{m/2}{\ceil{m/n}}\geq \frac n4
\]
columns.

Suppose instead that a saddlepoint, if it exists, must be at most $s$ (outcome (ii) of the test). Delete every row whose chosen entry satisfies $d_i > s$. Such a row cannot contain the saddlepoint, since the maximum in this row is at least $d_i>s$. At least $m/2-1\geq m/4$ rows are deleted.

If $n>m$, we apply the result to the transpose matrix with the order of the entries reversed (written $-A^\mathsf{T}$ when the entries are numeric). This exchanges the roles of rows and columns while preserving saddlepoints.
\end{proof}

\section{Reducing wide matrices}\label{sec:thin}

We now give the main new ingredient. The following certificate justifies the column deletions.

\begin{lemma}\label{lem:certificate}
Let $A$ be an $h\times w$ matrix with distinct entries. Fix a row $r\in[h]$ and a nonempty set of columns $C$, and let $c^*\in C$ maximize $A_{r,c}$ over $c\in C$. If every other row contains an entry larger than $A_{r,c^*}$ (not necessarily in a column of $C$), then no column in $C\setminus\{c^*\}$ contains a saddlepoint.
\end{lemma}

\begin{proof}
Fix $c\in C\setminus\{c^*\}$. The entry in row $r$ and column $c$ is not the row maximum, since $A_{r,c}<A_{r,c^*}$. Now consider a different row $i\ne r$. If $A_{i,c}$ were a saddlepoint, then it would be strictly smaller than $A_{r,c}$ because it is the strict minimum in column $c$. Hence
\[
A_{i,c}<A_{r,c}<A_{r,c^*}.
\]
But row $i$ contains an entry larger than $A_{r,c^*}$, so $A_{i,c}$ is not the maximum of row $i$, a contradiction.
\end{proof}

\begin{theorem}\label{thm:wide-reduction}
Let $\kappa$ be the constant from Lemma~\ref{lem:list-selection}. Given an $h\times w$ matrix $A$ with distinct entries and $w\geq 6\kappa h$, we can find, in $O(h+w)$ time, a set of at least $\ceil{w/(4\kappa)}$ columns, none of which contains a saddlepoint of $A$.
\end{theorem}

\begin{proof}
We construct one nondecreasing list $L_r$ for each row $r\in[h]$. Across all lists, no two items are assigned the same matrix column. It also associates with each item a witness column for the largest entry seen so far in that list.

\paragraph{The lists.}
For row $r$, let $c_{r,1}<c_{r,2}<\ldots$ be the columns assigned to successive items of $L_r$. Define
\[
p_{r,j}=\max_{1\leq \ell\leq j} A_{r,c_{r,\ell}},
\]
and let $z_{r,j}\in\{c_{r,1},\ldots,c_{r,j}\}$ be the unique column attaining this maximum. The $j$th list item is
\[
L_r[j]=(p_{r,j},c_{r,j}),
\]
ordered lexicographically. We write
\[
\key(L_r[j])=p_{r,j},\qquad
\col(L_r[j])=c_{r,j},\qquad
\wit(L_r[j])=z_{r,j}.
\]
Because the prefix maxima $p_{r,j}$ never decrease, and because the assigned column indices increase within each list whenever the key stays unchanged, every list is nondecreasing.

We generate the lists lazily as they are queried from left to right. Initially, set $c_{r,1}=r$ (i.e.\ $L_r[1]=(A_{r,r},r)$), for every $r\in[h]$; this is possible because $w\geq h$. For $j\geq2$, when querying $L_r[j]$ for the first time, we assign the smallest column $c$ not assigned previously to any item and set $c_{r,j}=c$. Thus all assigned columns across all lists are distinct, and at every moment they form an initial segment $[m]$ of the columns for some $m$. Moreover, $p_{r,j}$ can be computed in constant time from $p_{r,j-1}$ and $A_{r,c_{r,j}}$, as $p_{r,j} = \max\{p_{r,j-1}, A_{r,c_{r,j}}\}$.

This lazy generation is compatible with Lemma~\ref{lem:list-selection}: the algorithm reads every list from left to right, and every answer is consistent with a nondecreasing list. The output of the selection algorithm is the same as if it were given the fully constructed lists in advance (the algorithm does not ``see'' that the lists are lazily generated). The query bound below shows that, in our use case, the lazy generation never runs out of columns to assign to queried items.

\paragraph{The selection.}
Apply Lemma~\ref{lem:list-selection} to the $h$ lists and select
\[
K=h+\ceil{\frac{w}{4\kappa}}
\]
items. The selection algorithm queries at most $\kappa(h+K)$ items. We then additionally query the first unselected item in every list (if it has not already been queried by the selection algorithm), using at most $h$ additional queries. The total number of queried items is at most
\begin{align*}
\kappa(h+K)+h 
&\leq \kappa\left(2h+\frac{w}{4\kappa}+1\right)+h\\
&\leq \kappa\left(3h+\frac{w}{4\kappa}+1\right)\\
&\leq \frac{3}{4}w+\kappa\\
&\leq \frac{11}{12}w
< w.
\end{align*}
Here we used $\kappa \geq 1$, $w\geq 6\kappa h$ and $h\geq1$, which also imply $w\geq6\kappa$. Hence an unused column is available for every query, and all queried items are generated in $O(h+w)$ total time.

\paragraph{The deletions.}
Let $t_r$ be the number of selected items from $L_r$. Since $L_r$ is nondecreasing, these items form the prefix
\[
L_r[1],\ldots,L_r[t_r],
\]
and
\[
\sum_{r=1}^h t_r=K.
\]
Call row $r$ \emph{active} if $t_r>0$. For an active row $r$, let
\[
C_r=\{\col(L_r[1]),\ldots,\col(L_r[t_r])\},
\qquad
c_r^*=\wit(L_r[t_r]).
\]
By construction, $c_r^*\in C_r$ and
\[
A_{r,c_r^*}=\max_{c\in C_r} A_{r,c}.
\]

Fix an active row $r$, and write $q_r=A_{r,c_r^*}$. Consider any other row $i\ne r$. The first unselected item $L_i[t_i+1]$ is at least as large as every selected item in the lexicographic order. Therefore its key is at least $q_r$. These two keys are values of matrix entries in different rows, and all entries of $A$ are distinct, so equality is impossible. Hence
\[
\key(L_i[t_i+1])>q_r.
\]
Thus every row other than $r$ contains an entry larger than $A_{r,c_r^*}$. By Lemma~\ref{lem:certificate}, every column in $C_r\setminus\{c_r^*\}$ can be deleted.

The sets $C_r$ are pairwise disjoint. If $a$ is the number of active rows, then $a\leq h$, and the total number of deletable columns is
\[
\sum_{r:t_r>0} (t_r-1)
=K-a
\geq \ceil{\frac{w}{4\kappa}}.
\]
This proves the theorem.
\end{proof}

Applying Theorem~\ref{thm:wide-reduction} to $A^\mathsf{T}$ with the order of the entries reversed (or to $-A^\mathsf{T}$ when the entries are numeric) gives the symmetric statement.

\begin{corollary}\label{cor:wide-symmetric}
Given an $h\times w$ matrix $A$ with distinct entries and $h\geq 6\kappa w$, we can find, in $O(h+w)$ time, a set of at least $\ceil{h/(4\kappa)}$ rows, none of which contains a saddlepoint of $A$.
\end{corollary}

\section{The linear-time algorithm}\label{sec:overall}

We can now prove the main theorem.

\begin{proof}[Proof of Theorem~\ref{thm:main}]
Let $\kappa$ be the constant from Lemma~\ref{lem:list-selection}, and let
\[
\alpha = 1-\frac{1}{4\kappa}.
\]

We assume that $n$ is larger than a suitable absolute constant (smaller matrices can be handled directly by brute force). Recall that, without loss of generality, the entries of $A$ are distinct.

We first describe one round of the algorithm and show that it reduces both dimensions to at most $\alpha n$ in $O(n)$ time. Let
\[
t=\ceil{\frac{\ln(6\kappa)}{\ln(4/3)}},
\]
and apply Lemma~\ref{lem:one-sided} $t$ times. Since $t$ is a constant and every intermediate matrix has at most $n$ rows and $n$ columns, this takes $O(n)$ time. There are two cases.

\smallskip
\noindent\textbf{Case 1: both dimensions are reduced.}
If at least one application deletes rows and at least one application deletes columns, then both dimensions are at most $3n/4$. Since $\kappa\geq1$, we have $3/4\leq\alpha$, so both dimensions are at most $\alpha n$.

\smallskip
\noindent\textbf{Case 2: only one dimension is reduced.}
Suppose, without loss of generality, that every application deletes rows. The height becomes at most
\[
\left(\frac{3}{4}\right)^t n\leq \frac{n}{6\kappa},
\]
while the width remains $n$. Hence the resulting $h\times w$ matrix satisfies $w=n\geq 6\kappa h$. By Theorem~\ref{thm:wide-reduction}, we can delete at least $\ceil{w/(4\kappa)}$ columns in $O(h+w)=O(n)$ time. The new width is at most
\[
w-\ceil{\frac{w}{4\kappa}}
\leq \left(1-\frac{1}{4\kappa}\right)w
=\alpha n.
\]
Since $h\leq n/(6\kappa)\leq\alpha n$, the height is also at most $\alpha n$. The case in which every application deletes columns is symmetric, using Corollary~\ref{cor:wide-symmetric}.

Thus one round produces a submatrix whose two dimensions are at most $\alpha n$, in $O(n)$ time, while preserving the original saddlepoint if one exists. If necessary, restore deleted rows or columns on the shorter side until the matrix is square. Its side length is still at most $\alpha n$. Then repeat the round.

When the matrix has constant size, solve it directly. Alternatively, one can stop already once the matrix has size $n/\log n \times n/\log n$ and apply the $O(N\log N)$-time algorithm from \cite{BienstockEtAl,ByrneVaserstein}. If the reduced matrix has no saddlepoint, then the original matrix had none. Otherwise, let $s$ be its saddlepoint and verify $s$ in the original matrix $A$: if $s$ is a saddlepoint of $A$, we return it; otherwise, $A$ has no strict saddlepoint.

The running time satisfies
\[
T(n)\leq T(\alpha n)+O(n),
\]
so $T(n)=O(n)$.
\end{proof}

\section{Concluding remarks}

The algorithm closes the deterministic complexity gap for strict saddlepoints: randomization is not needed to match the natural linear-time lower bound. As in the earlier work, the result immediately extends to rectangular matrices by a simple reduction to square matrices~\cite{LlewellynToveyTrick,BienstockEtAl}: given an $m\times n$ matrix $A$, we can identify its strict saddlepoint, or report that none exists, deterministically in $O(m+n)$ time.

The only non-elementary black box in our algorithm is linear-time selection from sorted lists, which can itself be implemented using the simplified soft-heap approach of Kaplan et al.~\cite{KaplanEtAl}.

\bibliographystyle{plainurl}
{
\bibliography{references}
}

\end{document}